\documentclass[journal]{new-aiaa}

\usepackage{amsmath,amssymb,amsthm,mathtools,bm}
\usepackage{booktabs,tabularx,array,multirow}
\usepackage{graphicx}
\usepackage{placeins}
\usepackage{xcolor}

\hypersetup{
  hidelinks,
  pdftitle={Compositional Aeroelastic Operators for Morphing Flexible Multibody Aircraft: A Geometric Framework with Structural Verification},
  pdfauthor={Gelin Chen, Chen Song, and Chao Yang},
  pdfsubject={AIAA-style journal manuscript on compositional morphing-aircraft aeroelasticity},
  pdfkeywords={morphing aircraft, geometrically exact beam, Euler-Poincare equations, aeroelasticity, material attachment, model reduction}
}

\graphicspath{{figures/}}

\newcommand{\SE}{\mathrm{SE}(3)}

\newcommand{\Exp}{\operatorname{Exp}}
\newcommand{\Log}{\operatorname{Log}}
\newcommand{\Ad}{\operatorname{Ad}}
\newcommand{\ad}{\operatorname{ad}}
\newcommand{\dexp}{\operatorname{dexp}}
\newcommand{\BCH}{\operatorname{BCH}}

\newcommand{\diag}{\operatorname{diag}}
\newcommand{\R}{\mathbb{R}}
\newcommand{\PiN}[1]{\Pi_{\leq #1}}
\newcommand{\dd}{\,\mathrm{d}}
\newcommand{\T}{^{\mathsf T}}
\newcommand{\norm}[1]{\left\lVert #1\right\rVert}
\newcommand{\pair}[2]{\left\langle #1,#2\right\rangle}
\newcommand{\SI}[2]{\ensuremath{#1\,\mathrm{#2}}}

\newtheorem{proposition}{Proposition}

\title{Compositional Aeroelastic Operators for Morphing Flexible Multibody Aircraft:\\
A Geometric Framework with Structural Verification}

\author[1]{Gelin Chen\thanks{Doctoral Candidate, School of Aeronautic Science and Engineering; sy2405230@buaa.edu.cn.}}
\author[1]{Chen Song\thanks{Associate Professor, School of Aeronautic Science and Engineering; songchen@buaa.edu.cn.}}
\author[1]{Chao Yang\thanks{Associate Professor, School of Aeronautic Science and Engineering; yangchao@buaa.edu.cn.}}
\affil[1]{School of Aeronautic Science and Engineering, Beihang University, Beijing, People's Republic of China}

\begin{document}

\maketitle

\begin{abstract}
Morphing flexible multibody aircraft require structural strain, aerodynamic geometry, surface velocity, and generalized loading to remain compatible while joints and flexible components change configuration. A compositional formulation is developed around an assumed material attachment between each lifting surface and a geometrically exact beam. The component root pose is separated from the section field. Consequently, the body strain $G^{-1}G_{,s}$ and the elastic potential of a component depend on that component's elastic coordinates, while upstream motion enters its kinetic terms and external-load pullbacks. At element level, an exact relative logarithm $d$ supplies strain and potential energy, whereas a reference-anchored section coordinate $\sigma$ supplies the deformed section geometry. Their finite-order expansions retain the finite reference geometry exactly and truncate only endpoint perturbations. A fixed attachment map then generates surface points, tangents, normals, velocities, and force Jacobians from the same section kinematics. Euler--Poincar\'e beam balance, boundary ports, moving-surface potential-flow relations, graph cotangent assembly, and the associated semidiscrete power identity are stated in a common twist--wrench convention. Collocation, pressure evaluation, equivalent load application, and structural station projection are defined separately, exposing their distinct approximation errors. Verification shows the expected $N_d+1$ convergence order for degree-$N_d$ relative-log expansions. A geometrically nonlinear cantilever comparison further shows that a cubic static-manifold correction reduces mean full-record displacement error from $0.479$ to $0.255$ over four completed load cases. The reported evidence is structural and interface level; it does not constitute validation of a complete aircraft aeroelastic prediction. It supports the local geometric, algebraic, and virtual-work consistency of the construction and provides a traceable basis for subsequent coupled validation.
\end{abstract}

\section*{Nomenclature}

\noindent\begin{tabularx}{\textwidth}{@{}>{\raggedright\arraybackslash}p{0.19\textwidth}X@{}}
$H_{AB}$ & pose mapping coordinates from frame $B$ to frame $A$ \\
$B_i$ & root pose of flexible component $i$ \\
$G_i(s,t)$ & material-section pose of component $i$ \\
$G_i^0(s)$ & reference material-section pose \\
$s$ & material coordinate along a beam component \\
$q_i$ & retained elastic generalized coordinates of component $i$ \\
$x_i$ & full constraint-compatible structural coordinates before condensation \\
$\xi,\kappa$ & body twist and body strain, $(G^{-1}\dot G)^\vee$ and $(G^{-1}G_{,s})^\vee$ \\
$\epsilon$ & elastic strain measure $\kappa-\kappa^0$ \\
$\mu,n$ & sectional momentum and stress-resultant covectors \\
$d_0,d,\delta d$ & reference relative log, current relative log, and $d-d_0$ \\
$\sigma$ & section coordinate relative to the reference interpolation \\
$\Phi,\Psi$ & retained and condensed structural bases \\
$h_n$ & degree-$n$ term of the static condensation manifold \\
$(u,v)$ & oriented aerodynamic surface parameters; leading edge to trailing edge and root to tip \\
$\psi:(u,v)\mapsto(s,\upsilon)$ & locally invertible material attachment-coordinate map on a surface patch \\
$\rho(s,\upsilon)$ & surface offset expressed in the attached reference section frame \\
$X,n_s,A_s$ & current surface point, oriented unit normal, and surface area measure \\
$X_O,X_P,X_L,X_{\widetilde L}$ & operator, pressure, equivalent-load, and structural station-proxy sites \\
$\Gamma,x_w$ & bound-circulation and wake-state vectors \\
$\pi$ & pressure jump divided by density, $\Delta p/\rho_f$ \\
$b_a,y$ & aerodynamic boundary-condition input and selected generalized-force output \\
$\Pi_{\leq N}$ & projection onto multivariate polynomials of total degree at most $N$ \\
$\Ad_H,\ad_\xi$ & group and algebra adjoint operators \\
$\pair{w}{V}=w\T V$ & wrench--twist duality pairing \\
\end{tabularx}

\section{Introduction}

Morphing aircraft embed a changing geometry and load path inside the aeroelastic problem. Joint rotations and compliant shape changes alter the placement of lifting surfaces, the velocity entering a no-penetration condition, and the virtual directions through which pressure performs work. Reviews of morphing concepts identify this interaction among geometry, actuation, loads, and stability as a persistent design difficulty \cite{Sofla2010,Barbarino2011}. Recent aeroelastic studies have treated compliant span morphing and morphing flying wings with increasing structural and aerodynamic fidelity \cite{Ajaj2018,Syed2022}. Their breadth reinforces a basic requirement: a useful reduced model must retain the relation between structural material motion and the aerodynamic surface throughout a configuration change.

Geometrically exact beams provide invariant strain measures for large motion \cite{Simo1985,SimoVuQuoc1986,Hodges1990}. Euler--Poincar\'e reduction supplies compatible velocity, variation, momentum, and boundary-port equations on Lie groups \cite{Holm1998,Muller2018}. Unsteady vortex-lattice methods provide economical potential-flow models for flexible-aircraft dynamics \cite{Murua2012,KatzPlotkin2001}, while wake-state realization and balancing can reduce their dynamic order \cite{Maraniello2019}. Nonlinear flexible-aircraft formulations demonstrate how these ingredients can coexist in a single simulation framework \cite{Palacios2010,Hesse2014}. The remaining issue addressed here is the organization of their geometric interfaces.

General fluid--structure transfer methods emphasize work and momentum conservation on nonmatching meshes \cite{Farhat1998}. Radial-basis-function methods also offer effective and flexible interpolation for both displacement transfer and mesh motion \cite{Rendall2008}. The present work considers a more specialized hypothesis: each aerodynamic material point is assigned to a beam material section and to an offset expressed in that section's reference frame. Under this hypothesis, the surface point, tangents, oriented normal, point velocity, and generalized force map can all be differentiated from one attachment relation. This common source provides geometric consistency within the assumed attachment model. Comparative accuracy with other transfer families depends on the structural idealization, surface resolution, deformation class, and calibration data, and is therefore left to dedicated studies.

The main contributions are as follows. First, the formulation distinguishes the element relative logarithm $d$, which generates structural strain energy, from the section coordinate $\sigma$, which generates internal section geometry. Both are expanded about a finite reference geometry without counting that reference as a small elastic quantity. Second, the component factorization $G_i=B_i\overline G_i$ is used to show that the strain energy of component $i$ is $U_i(q_i)$ when $B_i$ is independent of $s$. Upstream coordinates still affect inertia, external work, and graph ports. Third, an assumed attachment map generates the complete aerodynamic geometry and its tangent and cotangent actions. Fourth, the Euler--Poincar\'e field equation, aerodynamic boundary and pressure equations, and multibody graph assembly are written using one twist--wrench pairing. Finally, exact-reference convergence data and a nonlinear cantilever comparison provide structural verification, including a direct comparison of quadratic and cubic static-manifold corrections.

The paper is organized around the theoretical dependency chain. Section~\ref{sec:setting} establishes the geometric hypotheses and graph convention. Section~\ref{sec:element} develops $d$ and $\sigma$. Section~\ref{sec:structure} states the Euler--Poincar\'e equation, component-local potential, and nonlinear static condensation. Sections~\ref{sec:attachment} and \ref{sec:aero} derive the surface geometry and aerodynamic one-forms. Section~\ref{sec:graph} gives cotangent graph assembly and output-aware wake reduction. Section~\ref{sec:evidence} reports verification evidence. Error sources and scope are discussed in Section~\ref{sec:scope}; derivations and a three-segment specialization are collected in the appendices.

\section{Geometric Setting and Compositional Hypotheses}\label{sec:setting}

\subsection{Twists, wrenches, and the component graph}

Let $H_{AB}\in\SE$ map coordinates expressed in frame $B$ to frame $A$. Twists and wrenches are ordered as
\begin{equation}
 V=\begin{bmatrix}v\\\omega\end{bmatrix},\qquad
 w=\begin{bmatrix}f\\m\end{bmatrix},\qquad
 \pair{w}{V}=w\T V .
 \label{eq:pairing}
\end{equation}
Coordinate changes obey $V_A=\Ad_{H_{AB}}V_B$ and $w_B=\Ad_{H_{AB}}\T w_A$. Thus, every forward tangent map has a uniquely paired cotangent map given by its transpose under Eq.~\eqref{eq:pairing}.

The aircraft is represented by a directed acyclic component graph. Flexible component $i$ has a parent port pose $P_i$, a joint pose $C_i(\theta_i)$, and a fixed installation pose $H_i^{\mathrm{inst}}$. Its root pose is
\begin{equation}
 B_i=P_i C_i(\theta_i)H_i^{\mathrm{inst}}.
 \label{eq:root_pose}
\end{equation}
If $J_i(\theta_i)$ is defined by $C_i^{-1}\dot C_i=(J_i\dot\theta_i)^\wedge$, the root twist has the recursion
\begin{equation}
 V_{B_i}=\Ad_{(H_i^{\mathrm{inst}})^{-1}C_i^{-1}}V_{P_i}
 +\Ad_{(H_i^{\mathrm{inst}})^{-1}}J_i\dot\theta_i .
 \label{eq:root_twist}
\end{equation}
The same linear maps act on admissible virtual twists. Their transposes therefore give the compatible backward wrench recursion; no separate moment-transfer convention is required.

\subsection{Structural and attachment hypotheses}

The material-section pose of component $i$ is factorized as
\begin{equation}
 G_i(s,t)=B_i(t)\overline G_i(s,q_i(t)),\qquad s\in[0,L_i],
 \label{eq:component_factor}
\end{equation}
where $B_i$ has no dependence on the material coordinate $s$. This factorization permits arbitrary upstream rigid and joint motion while isolating the component's elastic section field.

The aerodynamic surface carried by component $i$ is supplied in a reference component frame as $X_{i0}^{C}(u,v)$. On a surface patch $\Omega_i$, let
\begin{equation}
 \psi_i:\Omega_i\longrightarrow\widehat\Omega_i,qquad
 (u,v)\longmapsto\bigl(s_i(u,v),\upsilon_i(u,v)\bigr),
 \label{eq:attachment_chart}
\end{equation}
be a $C^1$ local diffeomorphism. Its material attachment is described by
\begin{equation}
 \mathcal A_i:(u,v)\longmapsto
 \bigl(s_i(u,v),\upsilon_i(u,v),
 \rho_i(s_i,\upsilon_i)\bigr),
 \label{eq:attachment_map}
\end{equation}
where $u$ increases from leading edge to trailing edge, $v$ increases from root to tip, and $\rho_i$ is an offset expressed in the attached section frame. If a global one-to-one chart is unavailable, the surface is covered by patches with compatible transition data. This condition makes $\rho_i(s,\upsilon)$ single valued on each patch; equivalently, one may store the composed field $\rho_i\circ\psi_i$ directly on $(u,v)$. The map is fixed in time for the principal development. Time-dependent or sliding attachment is discussed in Section~\ref{sec:scope}.

Three hypotheses delimit the theory:
\begin{enumerate}
 \item each flexible component admits Eq.~\eqref{eq:component_factor}, with its elastic constitutive data expressed in material coordinates;
 \item the attachment chart in Eq.~\eqref{eq:attachment_chart} is locally invertible on each surface patch, $\rho_i$ is single valued on its image, and the oriented surface parameterization is nondegenerate;
 \item local physical loads are represented as one-forms and are assembled through cotangent maps paired with the declared kinematics.
\end{enumerate}
These hypotheses describe an abstract relationship between structure and aerodynamics. Beam discretization, aerodynamic paneling, and state reduction enter afterward.

\subsection{Dependency chain}

For each structural element, endpoint coordinates determine the relative log $d$ and the section field $\sigma$. The structural branch of the chain is
\begin{equation}
 (q_a,q_b)\longrightarrow d\longrightarrow \epsilon\longrightarrow U_e,
 \label{eq:d_chain}
\end{equation}
whereas the geometric branch is
\begin{equation}
 (q_a,q_b)\longrightarrow \sigma\longrightarrow G(s)
 \longrightarrow \{X,X_{,u},X_{,v},n_s,\dot X,J_X\}.
 \label{eq:sigma_chain}
\end{equation}
The distinction in Eqs.~\eqref{eq:d_chain} and \eqref{eq:sigma_chain} prevents a section interpolation coordinate from silently becoming a strain measure. Their shared endpoint source also ensures that structural energy and surface geometry refer to the same deformed element.

\section{Element-Local Lie-Group Coordinates}\label{sec:element}

\subsection{Exact relative log and elastic potential}

Consider an element with reference endpoint poses $G_a^0$ and $G_b^0$. Endpoint perturbations $q_a,q_b\in\R^6$ define
\begin{equation}
 G_a=G_a^0\Exp(q_a^\wedge),\qquad
 G_b=G_b^0\Exp(q_b^\wedge),
 \label{eq:endpoint_poses}
\end{equation}
and the finite reference relative log is
\begin{equation}
 d_0=\Log\!\left((G_a^0)^{-1}G_b^0\right)^\vee .
 \label{eq:d0}
\end{equation}
The current relative log follows exactly from the group product,
\begin{equation}
 d(q_a,q_b)=
 \Log\!\left[\Exp(-q_a^\wedge)\Exp(d_0^\wedge)\Exp(q_b^\wedge)\right]^\vee,
 \qquad \delta d=d-d_0 .
 \label{eq:d_exact}
\end{equation}
For the exponential interpolation used below, the left-trivialized strain is constant on the element. Let
\begin{equation}
 \overline{\mathcal K}_e=\frac{1}{L_e}
 \int_{s_a}^{s_b}\mathcal K(s)\dd s
 \label{eq:average_section_stiffness}
\end{equation}
denote the element-average sectional constitutive matrix. The corresponding strain increment and potential are
\begin{equation}
 \epsilon_e=\frac{\delta d}{L_e},\qquad
 U_e=\frac{1}{2L_e}\delta d\T
 \overline{\mathcal K}_e\delta d,
 \label{eq:element_potential}
\end{equation}
which follows from $U_e=\tfrac12\int\epsilon_e\T\mathcal K(s)\epsilon_e\dd s$. Thus $\overline{\mathcal K}_e$ is an averaged sectional stiffness, not an already assembled element stiffness. For nonconstant strain interpolation or quadrature, the local field $\kappa(s)-\kappa^0(s)$ replaces $\delta d/L_e$ inside the integral. The role of $d$ is then to generate the compatible relative deformation from which that field is reconstructed.

The endpoint-relative chart
\begin{equation}
 A_0=\Ad_{\Exp(-d_0^\wedge)},\qquad q_b=A_0q_a+r
 \label{eq:endpoint_relative_chart}
\end{equation}
separates an approximately common endpoint perturbation from the relative increment $r$. It is useful for local finite-order construction because both $q_a$ and $r$ vanish at the element reference state.

\subsection{Section interpolation and the role of \texorpdfstring{$\sigma$}{sigma}}

Let $\lambda=(s-s_a)/L_e\in[0,1]$. The reference and current section interpolations are
\begin{align}
 G_e^0(\lambda)&=G_a^0\Exp(\lambda d_0^\wedge),
 \label{eq:reference_interp}\\
 G_e(\lambda)&=G_a^0\Exp(q_a^\wedge)\Exp(\lambda d^\wedge).
 \label{eq:current_interp}
\end{align}
The section coordinate relative to Eq.~\eqref{eq:reference_interp} is therefore
\begin{equation}
 \sigma(\lambda,q_a,r)=
 \Log\!\left[\Exp(-\lambda d_0^\wedge)
 \Exp(q_a^\wedge)\Exp(\lambda d^\wedge)\right]^\vee .
 \label{eq:sigma_exact}
\end{equation}
It satisfies $\sigma(0)=q_a$ and $\sigma(1)=q_b$ in the exact chart. The section pose may equivalently be written $G_e=G_e^0\Exp(\sigma^\wedge)$. Differentiation gives the operators used in the kinematic branch,
\begin{equation}
 E_\sigma=\Ad_{\Exp(-\sigma^\wedge)},\qquad
 D_\sigma=\dexp_{-\sigma},\qquad
 Y_\sigma=D_\sigma\sigma_{,q}.
 \label{eq:sigma_operators}
\end{equation}
The derivative $\sigma_{,s}=L_e^{-1}\sigma_{,\lambda}$ supplies section interpolation derivatives. The element potential in Eq.~\eqref{eq:element_potential} remains a function of $d$; $\sigma$ supplies section pose, velocity, virtual motion, and attached geometry.

\subsection{Reference-anchored finite-order expansions}

Ordinary Taylor--BCH truncation can inadvertently count the finite reference quantity $d_0$ as an elastic perturbation. A reference-anchored construction avoids that interpretation. For Eq.~\eqref{eq:sigma_exact}, define
\begin{equation}
 a=\lambda d_0,\qquad e=\lambda\delta d,\qquad
 \widehat q_a=\Ad_{\Exp(-a^\wedge)}q_a,
 \label{eq:anchored_defs}
\end{equation}
and
\begin{equation}
 \chi_R(a,e)=\Log\!\left[\Exp(-a^\wedge)\Exp((a+e)^\wedge)\right]^\vee .
 \label{eq:chiR}
\end{equation}
The exact factorization is
\begin{equation}
 \sigma=\BCH(\widehat q_a,\chi_R).
 \label{eq:sigma_bch}
\end{equation}
Here $a$ is retained as a finite anchor. The interaction-picture field
\begin{equation}
 b_I(\tau)=\Ad_{\Exp(-\tau a^\wedge)}e
 \label{eq:interaction_field}
\end{equation}
generates the Magnus series for $\chi_R$ \cite{Magnus1954,Blanes2009}. Its first terms are
\begin{align}
 \Omega_1&=\int_0^1 b_I(\tau_1)\dd\tau_1=\dexp_{-a}e,
 \label{eq:magnus1}\\
 \Omega_2&=\frac12\int_0^1\int_0^{\tau_1}
 [b_I(\tau_1),b_I(\tau_2)]\dd\tau_2\dd\tau_1,
 \label{eq:magnus2}
\end{align}
with $\chi_R=\sum_{k\geq1}\Omega_k$. Each $\Omega_k$ is degree $k$ in $e$, while dependence on $a$ remains exact. The same principle is applied to $d$ by inverting its anchored increment map; Appendix~\ref{app:anchored} gives the construction.

Let $z_e=(q_a,r)$. A degree-$N$ card is defined by applying $\Pi_{\leq N}$ after the anchored increment and small--small BCH operations:
\begin{equation}
 d^{[N]}=d_0+\PiN{N}\delta d(z_e),\qquad
 \sigma^{[N]}=\PiN{N}\sigma(z_e,\lambda).
 \label{eq:finite_cards}
\end{equation}
For a complete total-degree expansion in a valid local chart, the residual is $O(\norm{z_e}^{N+1})$. Finite-order claims therefore refer to endpoint perturbation degree, not to powers of the reference curvature or reference element rotation.

\section{Euler--Poincar\'e Structure and Nonlinear Reduction}\label{sec:structure}

\subsection{Body fields, compatibility, and the EP equation}

For a section pose $G(s,t)$, define
\begin{equation}
 \xi=(G^{-1}\dot G)^\vee,\qquad
 \kappa=(G^{-1}G_{,s})^\vee,
 \label{eq:xi_kappa}
\end{equation}
and a left-trivialized variation $\delta G=G\zeta^\wedge$. The compatibility identities are
\begin{equation}
 \delta\xi=\dot\zeta+\ad_\xi\zeta,\qquad
 \delta\kappa=\zeta_{,s}+\ad_\kappa\zeta .
 \label{eq:compatibility}
\end{equation}
Consider the sectional Lagrangian density
\begin{equation}
 \ell(\xi,\kappa;s)=\frac12\xi\T\mathcal M(s)\xi-W(\kappa;s),
 \quad
 \mu=\frac{\partial\ell}{\partial\xi},\quad
 n=\frac{\partial W}{\partial\kappa}.
 \label{eq:lag_density}
\end{equation}
The coadjoint is defined by $\pair{\ad_\xi^*\mu}{\zeta}=\pair{\mu}{\ad_\xi\zeta}$. With an applied distributed wrench density $f$, the d'Alembert form of the Euler--Poincar\'e beam equation is
\begin{equation}
 -\dot\mu+\ad_\xi^*\mu+n_{,s}-\ad_\kappa^*n+f=0.
 \label{eq:EP}
\end{equation}
Equation~\eqref{eq:EP} is retained as the structural field equation because it fixes the signs of inertia, elastic stress, external work, and boundary ports. At the endpoints, the outward internal ports are $-n(0)$ and $n(L)$. Appendix~\ref{app:EP} derives Eq.~\eqref{eq:EP} and the endpoint signs from Hamilton's principle.

For the factorization in Eq.~\eqref{eq:component_factor}, the section twist takes the form
\begin{equation}
 \xi_i=Z_i(s,q_i)V_{B_i}+Y_i(s,q_i)\dot q_i,
 \label{eq:section_twist_reduced}
\end{equation}
where the same maps act on root and elastic virtual velocities. Pullback of the weak form gives root and elastic residual one-forms. For a fixed root, the inertial contribution is
\begin{equation}
 R_{T,i}=\int_0^{L_i}Y_i\T
 \left(\ad_{\xi_i}^*\mu_i-\dot\mu_i\right)\dd s,
 \qquad \xi_i=Y_i\dot q_i.
 \label{eq:reduced_inertia}
\end{equation}
With $R_{U,i}=\partial U_i/\partial q_i$ and external generalized force $Q_i$, the d'Alembert residual is
\begin{equation}
 R_{T,i}-R_{U,i}+Q_i=0.
 \label{eq:reduced_residual}
\end{equation}
An algebraic solver residual may use the overall negative of Eq.~\eqref{eq:reduced_residual}; the representation must be chosen once so that external-force and potential signs are not mixed.

\subsection{Component-local elastic potential}

Let the body strain and elastic strain be
\begin{equation}
 \kappa_i=(G_i^{-1}G_{i,s})^\vee,\qquad
 \epsilon_i=\kappa_i-\kappa_i^0,
 \label{eq:strain_def}
\end{equation}
with potential
\begin{equation}
 U_i=\int_0^{L_i}W_i(\epsilon_i;s)\dd s,
 \qquad
 W_i=\frac12\epsilon_i\T\mathcal K_i(s)\epsilon_i
 \quad\text{for a linear sectional law.}
 \label{eq:continuum_potential}
\end{equation}

\begin{proposition}[Root invariance of component strain energy]\label{prop:root_invariance}
Under Eq.~\eqref{eq:component_factor}, with $B_{i,s}=0$ and material constitutive data independent of upstream configuration,
\begin{equation}
 \kappa_i=\left(\overline G_i^{-1}\overline G_{i,s}\right)^\vee,
 \qquad U_i=U_i(q_i).
 \label{eq:local_potential_result}
\end{equation}
Consequently, the direct derivative of $U_i$ with respect to an upstream root or joint coordinate is zero.
\end{proposition}

\begin{proof}
Because $G_i=B_i\overline G_i$ and $B_{i,s}=0$,
\begin{equation}
 G_i^{-1}G_{i,s}
 =\overline G_i^{-1}B_i^{-1}B_i\overline G_{i,s}
 =\overline G_i^{-1}\overline G_{i,s}.
 \label{eq:root_cancel}
\end{equation}
Substitution into Eq.~\eqref{eq:continuum_potential} removes $B_i$ and all coordinates that enter only through $B_i$.
\end{proof}

Proposition~\ref{prop:root_invariance} concerns the component's strain potential. Upstream coordinates still enter $\xi_i$ through Eq.~\eqref{eq:section_twist_reduced}, and external loads attached to the moving component pull back through the graph. Local stress resultants also supply endpoint ports within the component weak form. Joint springs, gravity, contact, or other configuration-dependent potentials are additional terms and need not satisfy Eq.~\eqref{eq:local_potential_result}.

\subsection{Mechanics-informed static condensation}

Let $x_i$ denote full constraint-compatible structural coordinates for a component. Retained and condensed bases $\Phi_i$ and $\Psi_i$ define
\begin{equation}
 x_i(q_i,y_i)=\Phi_iq_i+\Psi_i y_i .
 \label{eq:master_slave}
\end{equation}
The reference state is assumed to be an equilibrium, and the basis split is chosen so that the linear condensed response has already been absorbed into the retained chart. With $K_i=\nabla_x^2U_i(0)$, this condition is
\begin{equation}
 \Psi_i\T K_i\Phi_i=0.
 \label{eq:linear_decoupling}
\end{equation}
It implies $h(0)=0$ and $Dh(0)=0$. If Eq.~\eqref{eq:linear_decoupling} is not satisfied, a linear term $h_1(q)$ must be retained or the bases must be redefined.

The static manifold is determined by vanishing condensed force,
\begin{equation}
 g_y(q,y)=\Psi_i\T\nabla_x U_i(\Phi_iq+\Psi_i y)=0,
 \label{eq:static_manifold_eq}
\end{equation}
and is expanded as
\begin{equation}
 y=h(q)=h_2(q)+h_3(q)+\cdots+h_{N_h}(q),
 \label{eq:h_expansion}
\end{equation}
where $h_n$ is homogeneous of degree $n$. If
\begin{equation}
 K_{yy}=\Psi_i\T\nabla_x^2U_i(0)\Psi_i
 \label{eq:Kyy}
\end{equation}
is nonsingular on the condensed subspace, each coefficient satisfies a homological equation
\begin{equation}
 K_{yy}h_n=-\left[g_y\left(q,\sum_{k=2}^{n-1}h_k(q)\right)\right]_n,
 \qquad n\geq2,
 \label{eq:homological}
\end{equation}
where $[\cdot]_n$ extracts the degree-$n$ part. This construction is related to quadratic manifolds and higher-order invariant-manifold parameterizations used in nonlinear structural reduction \cite{Rutzmoser2017,Vizzaccaro2021,Haller2016}. Here it is used as a static equilibrium closure rather than as a claim of invariant dynamics.

The manifold also defines the reduced kinetic terms. Let
\begin{equation}
 \mathcal X(q)=\Phi q+\Psi h(q),\qquad
 J_h(q)=D\mathcal X(q)=\Phi+\Psi Dh(q).
 \label{eq:manifold_tangent}
\end{equation}
The variationally complete reduced Lagrangian is
\begin{equation}
 L_r(q,\dot q)=
 T\!\left(\mathcal X(q),J_h(q)\dot q\right)
 -U\!\left(\mathcal X(q)\right).
 \label{eq:reduced_lagrangian_manifold}
\end{equation}
Variation of Eq.~\eqref{eq:reduced_lagrangian_manifold} generates the configuration-dependent reduced inertia and its convective terms. Modifying only the potential while retaining the linear kinetic map would not be a complete manifold reduction. The comparisons in Section~\ref{sec:evidence} use the complete tangent pullback in Eq.~\eqref{eq:reduced_lagrangian_manifold}.

Two truncations used in Section~\ref{sec:evidence} are
\begin{align}
 x_{H_2}(q)&=\Phi q+\Psi h_2(q),
 \label{eq:H2}\\
 x_{H_2+H_3}(q)&=\Phi q+\Psi\{h_2(q)+h_3(q)\}.
 \label{eq:H23}
\end{align}
If $\delta d=\sum_{k\geq1}d_k$ is homogeneous in $q$, a complete fourth-order $H_2$ potential requires $d_1,d_2,d_3$, and a complete sixth-order $H_2+H_3$ potential requires terms through $d_5$. Appendix~\ref{app:condensation} shows the degree bookkeeping. This matching prevents selected higher-order products from being retained without their same-degree companions.

\section{Surface Geometry from an Assumed Material Attachment}\label{sec:attachment}

\subsection{Reference anchoring and current surface}

Let the reference section pose in component coordinates be $(R_{S0}^C(s),p_{S0}^C(s))$. The surface offset associated with Eq.~\eqref{eq:attachment_map} is obtained from the reference surface:
\begin{equation}
 \rho(s(u,v),\upsilon(u,v))=
 (R_{S0}^C(s))\T\left[X_0^C(u,v)-p_{S0}^C(s)\right].
 \label{eq:rho_reference}
\end{equation}
This definition enforces the reference identity
\begin{equation}
 X_0^C(u,v)=p_{S0}^C(s)+R_{S0}^C(s)\rho(s,\upsilon).
 \label{eq:reference_identity}
\end{equation}
For the current section pose $(R_S^I(s,q),p_S^I(s,q))$ generated from $\sigma$, the attached surface is
\begin{equation}
 X^I(u,v;q)=p_S^I(s,q)+R_S^I(s,q)\rho(s,\upsilon),
 \qquad (s,\upsilon)=\psi(u,v).
 \label{eq:current_surface}
\end{equation}
Equations~\eqref{eq:rho_reference}--\eqref{eq:current_surface} are the central structure--aerodynamic abstraction. Reference geometry establishes the material offsets once; the deformed section field transports those offsets thereafter.

\subsection{Tangents, metric, normal, and surface velocity}

Define the derivatives with respect to attachment coordinates by
\begin{align}
 X_{,s}&=p_{S,s}+R_{S,s}\rho+R_S\rho_{,s},
 \label{eq:Xs}\\
 X_{,\upsilon}&=R_S\rho_{,\upsilon}.
 \label{eq:Xupsilon}
\end{align}
The parameter-space chain rule gives
\begin{equation}
 X_{,u}=X_{,s}s_{,u}+X_{,\upsilon}\upsilon_{,u},\qquad
 X_{,v}=X_{,s}s_{,v}+X_{,\upsilon}\upsilon_{,v}.
 \label{eq:surface_chain_rule}
\end{equation}
The covariant basis, metric, dual basis, and oriented unit normal are
\begin{align}
 a_1&=X_{,u},\quad a_2=X_{,v},\quad
 g_{\alpha\beta}=a_\alpha\T a_\beta,
 \label{eq:metric}\\
 a^\alpha&=g^{\alpha\beta}a_\beta,\qquad
 n_s=\frac{X_{,u}\times X_{,v}}
 {\norm{X_{,u}\times X_{,v}}},
 \label{eq:surface_normal}
\end{align}
where $[g^{\alpha\beta}]=[g_{\alpha\beta}]^{-1}$. The prescribed $(u,v)$ orientation fixes the normal sign. In particular, replacing $X_{,u}\times X_{,v}$ with $X_{,s}\times X_{,\upsilon}$ can reverse the normal when the attachment chart has a negative parameter Jacobian.

At a fixed material surface coordinate $(u,v)$, $s$, $\upsilon$, and $\rho$ have no explicit time derivative. If the attached section has translational and angular velocities $(v_S,\omega_S)$, then
\begin{equation}
 \dot X=v_S+\omega_S\times (R_S\rho).
 \label{eq:surface_velocity}
\end{equation}
Using the section kinematics generated from $\sigma$, Eq.~\eqref{eq:surface_velocity} can be written
\begin{equation}
 \dot X=J_{X,B}(u,v;q)V_B+J_{X,q}(u,v;q)\dot q.
 \label{eq:point_jacobian}
\end{equation}
The corresponding variation has the same maps,
\begin{equation}
 \delta X=J_{X,B}\zeta_B+J_{X,q}\delta q.
 \label{eq:point_variation}
\end{equation}
Consequently, a point force $f_X$ contributes the one-forms
\begin{equation}
 Q_B=J_{X,B}\T f_X,\qquad Q_q=J_{X,q}\T f_X,
 \label{eq:force_pullback}
\end{equation}
and satisfies $f_X\T\delta X=Q_B\T\zeta_B+Q_q\T\delta q$ identically.

\begin{proposition}[Common-source compatibility]\label{prop:common_source}
Suppose $d^{[N]}$ and $\sigma^{[N]}$ use the same endpoint chart, and all quantities in Eqs.~\eqref{eq:current_surface}--\eqref{eq:force_pullback} are differentiated from $\sigma^{[N]}$ and $\mathcal A$. Then the rate map and force pullback are dual under Eq.~\eqref{eq:pairing}; the reference surface identity is exact; and the tangent and normal orientation is determined by the declared attachment parameterization.
\end{proposition}

\begin{proof}
Equation~\eqref{eq:reference_identity} follows directly from Eq.~\eqref{eq:rho_reference}. Differentiation of one function $X(u,v;q)$ gives Eqs.~\eqref{eq:surface_chain_rule}, \eqref{eq:point_jacobian}, and \eqref{eq:point_variation}. Transposition of the latter produces Eq.~\eqref{eq:force_pullback}; hence virtual work is preserved by construction.
\end{proof}

\subsection{Discrete sites and their error forms}\label{sec:point_roles}

Aerodynamic and structural discretizations can use different sites even when they share Eq.~\eqref{eq:current_surface}. We distinguish four roles:
\begin{enumerate}
 \item $X_O$: the operator or collocation site used in the no-penetration equation;
 \item $X_P$: the pressure and area-quadrature site;
 \item $X_L$: the equivalent force application site;
 \item $X_{\widetilde L}$: a structural station proxy used to express the load in a section frame.
\end{enumerate}
These sites may coincide. Their separation describes a discretization policy and reveals the associated errors.

For a distributed panel load $t(X)$, define the exact resultant and moment about $O$ by
\begin{equation}
 f=\int_{A_p}t\dd A,\qquad
 m_O=\int_{A_p}(X-O)\times t\dd A.
 \label{eq:panel_resultants}
\end{equation}
A force applied at $X_L$ leaves a residual couple
\begin{equation}
 c_{\mathrm{res}}=m_O-(X_L-O)\times f.
 \label{eq:residual_couple}
\end{equation}
Omitting $c_{\mathrm{res}}$ is a force-equivalence approximation, distinct from pressure quadrature. If the force is subsequently referred to the station proxy $X_{\widetilde L}$, the change in moment is
\begin{equation}
 \Delta m_{\mathrm{stat}}=(X_L-X_{\widetilde L})\times f,\qquad
 \norm{\Delta m_{\mathrm{stat}}}
 \leq \norm{X_L-X_{\widetilde L}}\norm{f}.
 \label{eq:station_error}
\end{equation}
Finally, using a right-hand-side velocity at $X_R\neq X_O$ introduces the kinematic defect
\begin{equation}
 e_{\mathrm{rhs}}=n_O\T\left[\dot X(X_R)-\dot X(X_O)\right].
 \label{eq:rhs_proxy_error}
\end{equation}
Equations~\eqref{eq:residual_couple}--\eqref{eq:rhs_proxy_error} distinguish omitted couple, station moment, and boundary-condition proxy errors. They should not be combined into a single unnamed ``load-transfer error.''

\section{Aerodynamic Equations on the Generated Surface}\label{sec:aero}

\subsection{No penetration, circulation, and wake state}

Let $U_\infty$ be the incident velocity and let $u_\Gamma$ and $u_w$ denote velocities induced by bound and wake vorticity. The no-penetration equation at $X_O$ is
\begin{equation}
 n_O\T\left(U_\infty+u_\Gamma+u_w-\dot X_O\right)=0.
 \label{eq:no_penetration}
\end{equation}
For a structured vortex-ring discretization this becomes
\begin{equation}
 \mathcal A_\Gamma(X,n)\Gamma_b
\mathcal A_{\mathrm w}(X,n)x_w
=b_a(X,n,\dot X,U_\infty),
 \label{eq:AIC}
\end{equation}
with a specified bound-circulation orientation and trailing-edge Kutta condition. The equations do not require a classical quarter-chord vortex and three-quarter-chord collocation placement; the operator site is part of the declared panel scheme.

After linearization about a declared reference geometry and wake-convection law, an unsteady wake realization may be written in continuous form as
\begin{equation}
 \dot x_w=F_wx_w+G_w b_a+G_{\dot b}\dot b_a,
 \label{eq:wake_state}
\end{equation}
or in its time-discrete counterpart. Equation~\eqref{eq:wake_state} does not represent a general nonlinear free wake on a moving geometry. The state definition, linearization point, wake convection assumption, and Kutta map must be stated together because they determine which circulation histories are represented \cite{Murua2012,Maraniello2019}.

\subsection{Moving-surface pressure equation}

Let
\begin{equation}
 \Gamma(u,v,t)=\phi_{\mathrm{upper}}-\phi_{\mathrm{lower}},
 \qquad \Delta p=p_{\mathrm{lower}}-p_{\mathrm{upper}},
 \label{eq:jump_convention}
\end{equation}
and let $\rho_f$ be fluid density. Thus $\Gamma$ is a signed surface-potential jump; its panel values are identified with bound-circulation degrees of freedom under the declared vortex-ring orientation. On the generated surface, the surface gradient is
\begin{equation}
 \nabla_s\Gamma=a^\alpha\Gamma_{,\alpha}.
 \label{eq:surface_gradient}
\end{equation}
For a moving material surface coordinate, the Bernoulli pressure relation is represented by
\begin{equation}
 \pi:=\frac{\Delta p}{\rho_f}
 =\left.\dot\Gamma\right|_{u,v}
 +(\overline u-\dot X)\T\nabla_s\Gamma
 =\left.\dot\Gamma\right|_{u,v}+c^\alpha\Gamma_{,\alpha},
 \label{eq:bernoulli}
\end{equation}
where $\overline u$ is the selected mean tangential flow and
\begin{equation}
 c^\alpha=(\overline u-\dot X)\T a^\alpha .
 \label{eq:contravariant_velocity}
\end{equation}
Here $\overline u$ is the arithmetic mean of the limiting upper and lower tangential velocities under the potential-flow model, including the incident and selected induced contributions. Equation~\eqref{eq:bernoulli} is the difference of the two unsteady Bernoulli relations expressed at fixed material coordinates; it includes surface metric effects through $a^\alpha$. A structured discretization gives
\begin{equation}
 \bm\pi=\dot{\bm\Gamma}_b+
 \sum_{\alpha=1}^{2}\diag(\bm c^\alpha)D_\alpha\bm\Gamma_b,
 \label{eq:discrete_bernoulli}
\end{equation}
where $D_\alpha$ differentiates circulation in the chosen surface coordinates. Frozen-reference or orthogonal-grid forms are obtained by evaluating $a^\alpha$, $c^\alpha$, and $\dot X$ under the corresponding specialization. Such forms are partial linearizations of Eq.~\eqref{eq:bernoulli}; they do not represent its complete Fr\'echet derivative when circulation-induced velocities and moving geometry are both varied.

With the sign convention $\Delta p=p_{\mathrm{lower}}-p_{\mathrm{upper}}$, a pressure element performs the virtual work
\begin{equation}
 \delta W_a=\int_{A} \rho_f\pi n_s\T\delta X\dd A,
 \label{eq:aero_work}
\end{equation}
and generates the point force $\dd f=\rho_f\pi n_s\dd A$. Pullback through Eq.~\eqref{eq:point_variation} gives the structural one-form. This sequence uses the same attachment-generated $n_s$, $\dot X$, and $J_X$ in the aerodynamic boundary condition and in the load map.

\section{Graph Cotangent Assembly and Wake Reduction}\label{sec:graph}

\subsection{Local one-forms and graph pullback}

Let $z$ collect a base twist, joint rates, and component elastic rates. Forward kinematics from $z$ to a local section or surface point has a tangent map
\begin{equation}
 \delta x_\ell=J_\ell(z)\delta z.
 \label{eq:local_tangent}
\end{equation}
For a local one-form $\alpha_\ell$, graph assembly is the cotangent pullback
\begin{equation}
 \alpha_z=J_\ell(z)\T\alpha_\ell,
 \qquad
 \pair{\alpha_\ell}{\delta x_\ell}
 =\pair{\alpha_z}{\delta z}.
 \label{eq:graph_pullback}
\end{equation}
Summing Eq.~\eqref{eq:graph_pullback} over local inertial, elastic, and aerodynamic contributions yields the global d'Alembert one-form. In a tree, the backward recursion accumulates every descendant wrench at its parent joint. A proximal joint therefore receives the wrench of its entire downstream subtree, while a distal joint receives only its own subtree. Appendix~\ref{app:three_segment} gives the explicit three-segment formula.

\begin{proposition}[Semidiscrete power identity and internal-port cancellation]
\label{prop:power_identity}
Let $\dot z$ be an admissible generalized velocity and let every local load one-form be assembled by the transpose of the same tangent map used for local velocity. Then
\begin{equation}
 \dot z\T Q_{\mathrm{ext}}
 =\sum_\ell \pair{\alpha_\ell}{J_\ell\dot z}.
 \label{eq:semidiscrete_power}
\end{equation}
For the pressure load in Eq.~\eqref{eq:aero_work},
\begin{equation}
 \dot z\T Q_a
 =\int_A\rho_f\pi n_s\T\dot X\dd A.
 \label{eq:aero_power_identity}
\end{equation}
Equal-and-opposite wrenches at an internal graph port cancel from the sum. If the reduced mechanical Lagrangian has no explicit time dependence, a solution of the assembled d'Alembert equations therefore satisfies
\begin{equation}
 \frac{\dd}{\dd t}(T_r+U_r)
 =\dot z\T Q_{\mathrm{ext}},
 \label{eq:mechanical_energy_rate}
\end{equation}
up to the work of any explicitly declared dissipative or nonconservative terms.
\end{proposition}

\begin{proof}
Equation~\eqref{eq:semidiscrete_power} follows from $Q_{\mathrm{ext}}=\sum_\ell J_\ell\T\alpha_\ell$. Substitution of $\dot X=J_X\dot z$ gives Eq.~\eqref{eq:aero_power_identity}. At a connected port the two components have the same admissible port velocity and opposite wrenches, so their powers sum to zero. Multiplying the assembled Euler--Lagrange--d'Alembert equation by $\dot z$ gives Eq.~\eqref{eq:mechanical_energy_rate}.
\end{proof}

For an equivalent force $f$ and free couple $c_L$ at $X_L$, the wrench expressed in an attached section frame $S$ is
\begin{equation}
 w_S=
 \begin{bmatrix}
 R_S\T f\\
 R_S\T\{(X_L-X_S)\times f+c_L\}
 \end{bmatrix}.
 \label{eq:section_wrench}
\end{equation}
The lever arm in Eq.~\eqref{eq:section_wrench} and the graph adjoints in Eq.~\eqref{eq:graph_pullback} describe complementary operations: the former changes the wrench reference point within a component, whereas the latter pulls the resulting covector through generalized coordinates and joints.

\subsection{Mechanism-relevant aerodynamic outputs}

The aerodynamic output vector is selected after local pressure loads have been pulled to meaningful structural coordinates. A typical choice is
\begin{equation}
 y=S_Q Q_a,
 \label{eq:selected_outputs}
\end{equation}
where $Q_a$ contains component-root, elastic-mode, and joint generalized forces, and $S_Q$ selects the channels needed by the mechanism model. This ordering avoids reducing wake dynamics against panel pressure coordinates that may poorly represent joint or modal work.

For a linearized aerodynamic realization
\begin{align}
 \dot x_w&=A x_w+B b_a+B_{\dot b}\dot b_a,
 \label{eq:aero_ss1}\\
 y&=C x_w+D b_a+D_{\dot b}\dot b_a,
 \label{eq:aero_ss2}
\end{align}
balanced truncation or another projection-based method is applied only to $x_w$ \cite{Moore1981,Glover1984,Benner2015}. With right and left projectors $T_R,T_L$ satisfying $T_L\T T_R=I$, the reduced state matrices are
\begin{equation}
 A_r=T_L\T AT_R,\quad B_r=T_L\T B,\quad
 B_{\dot b,r}=T_L\T B_{\dot b},\quad C_r=CT_R.
 \label{eq:balanced_reduction}
\end{equation}
The direct channels remain
\begin{equation}
 D_r=D,\qquad D_{\dot b,r}=D_{\dot b}.
 \label{eq:direct_preservation}
\end{equation}
Preserving Eq.~\eqref{eq:direct_preservation} is important for moving-surface inputs because an input-rate term can carry instantaneous added-circulation or discretized Bernoulli effects. The selected output $C$ is a generalized-force map. A row norm or a ratio of its entries may be useful as a sensitivity indicator, yet it is not an observability Gramian. Observability terminology is reserved for a quantity formed from the state dynamics and output pair, such as the Gramian used in balanced truncation.

\section{Verification and Structural Evidence}\label{sec:evidence}

The evidence in this section addresses three questions: whether the anchored relative-log construction has the expected local order, whether the static manifold closes its condensed equilibrium equations, and whether the $H_2+H_3$ structural route changes the complete transient response in the expected direction relative to $H_2$. It is not a validation of a complete aircraft aeroelastic prediction.

\subsection{Exact-reference relative-log convergence}

Let $\overline z_e$ be a fixed normalized endpoint perturbation direction and define
\begin{equation}
 e_{d,N_d}(h)=
 \norm{\delta d^{[N_d]}(h\overline z_e)
 -\delta d_{\mathrm{exact}}(h\overline z_e)}_2 .
 \label{eq:d_error}
\end{equation}
The finite reference $d_0$ is held fixed while $h$ ranges from $0.12$ to $0.021$. Figure~\ref{fig:d_convergence} shows raw oracle errors for degrees $N_d=1,\ldots,5$. Least-squares slopes are $2.001$, $3.000$, $4.001$, $5.001$, and $6.000$, respectively. The observed $N_d+1$ behavior supports the perturbation-degree interpretation in Eq.~\eqref{eq:finite_cards}. Rigid-reference tests produced zero reported error, and directional-action discrepancies for the five degrees ranged from $6.6\times10^{-13}$ to $6.0\times10^{-12}$.

This is a local manufactured-oracle check along one stored perturbation direction. It does not by itself characterize chart-boundary robustness or establish uniform error constants over reference geometries. A broader verification campaign should sample multiple $d_0$, endpoint directions, and section locations and should report value and derivative convergence for both $d$ and $\sigma$.

\begin{figure}[htbp]
 \centering
 \includegraphics[width=0.84\textwidth]{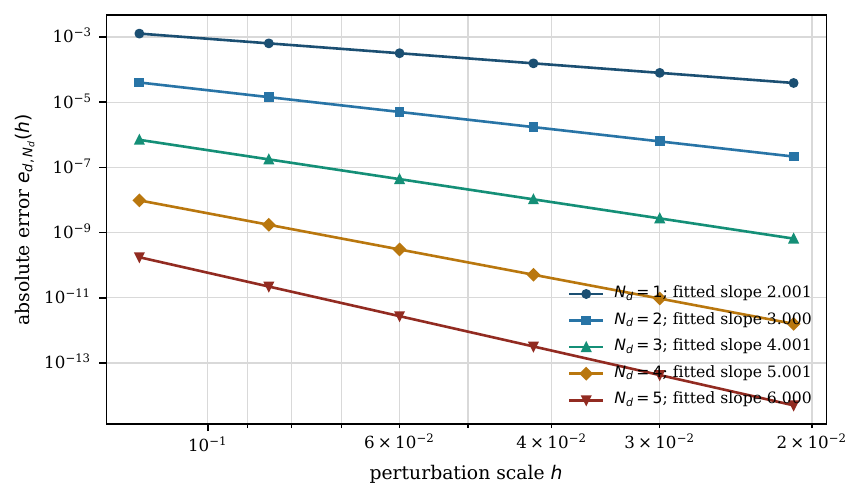}
 \caption{Exact-reference relative-log convergence. The plotted quantity is the absolute error $e_{d,N_d}(h)$ defined in Eq.~\eqref{eq:d_error}; the finite reference geometry is unchanged across perturbation scales.}
 \label{fig:d_convergence}
\end{figure}

\subsection{Static-manifold closure}

For the component reduction in Eqs.~\eqref{eq:static_manifold_eq}--\eqref{eq:homological}, define the degree-$n$ normalized closure
\begin{equation}
 c_n=\frac{\norm{[g_y(q,h_{\leq n})]_n}_2}
 {\max\left(\norm{K_{yy}h_n}_2,
 \norm{[g_y(q,h_{<n})]_n}_2,\epsilon_{\mathrm{mach}}\right)}.
 \label{eq:closure_metric}
\end{equation}
The stored coefficient residuals give $c_2=c_3=0$ to reported precision, with nonzero normalization denominators $1.1246\times10^{-2}$ and $5.6309\times10^{-5}$. The common fourth-order potential coefficients of the $H_2$ and $H_2+H_3$ constructions also agree to reported precision:
\begin{equation}
 \Pi_{\leq4}\overline U_{H_2+H_3}=\overline U_{H_2}.
 \label{eq:potential_consistency}
\end{equation}
These checks establish algebraic closure of the selected static-manifold orders and rule out a low-order mismatch between the two reduced potentials.

\subsection{Geometrically nonlinear cantilever comparison}

The structural example is a straight aluminum cantilever with properties listed in Table~\ref{tab:beam_setup}. The reference comparison uses a geometrically nonlinear Timoshenko beam model in COMSOL with the same beam dimensions, material properties, load history, and output definition. Three retained bending modes have reference frequencies $4.939$, $31.108$, and $87.813$~Hz. Fifteen condensed directions supply the static manifold. No damping is applied.

\begin{table}[htbp]
\caption{Cantilever configuration and explicitly defined loading parameter.}
\label{tab:beam_setup}
\centering
\begin{tabular}{@{}ll@{}}
\toprule
Quantity & Value \\
\midrule
Length $L$, width $b_s$, thickness $h_s$ & \SI{1.0}{m}, \SI{0.03}{m}, \SI{0.006}{m} \\
Young's modulus $E$, Poisson ratio $\nu$ & \SI{70}{GPa}, $0.33$ \\
Mass density $\rho_s$, shear factor $k_s$ & \SI{2700}{kg.m^{-3}}, $5/6$ \\
Structural discretization & 15 Timoshenko beam elements \\
Record length $T$ & \SI{2.5}{s} \\
Reference force $F_{\mathrm{ref}}$ & \SI{113.4}{N} \\
Load ratio & $\eta=F_0/F_{\mathrm{ref}}$ \\
Force history & $F(t)=F_0\sin(2\pi\,0.7f_1t)\min[t/(2/f_1),1]$ \\
Output $w(t)$ & tip displacement along the applied-force direction \\
\bottomrule
\end{tabular}
\end{table}

Four completed reference cases are available for $\eta\in\{0.10,0.15,0.20,0.25\}$. The $\eta=0.30$ COMSOL record is incomplete and is excluded from the reported means; its failure is not interpreted as a physical instability or as evidence for either reduced route. Resolving that case requires a separate reference-solver convergence study. For route $r$, two reported comparison measures are
\begin{equation}
 e_{2,r}=\frac{\norm{w_r-w_C}_2}{\norm{w_C}_2},\qquad
 e_{A,r}=\frac{|A_r-A_C|}{A_C},\qquad
 A_r=\max_{0\leq t\leq T}|w_r(t)|,
 \label{eq:beam_errors}
\end{equation}
where subscript $C$ denotes COMSOL and the discrete $2$-norm is taken on the common time grid. Thus $e_2$ measures the complete record, including phase and waveform, while $e_A$ measures peak magnitude only.

\begin{figure}[htbp]
 \centering
 \includegraphics[width=0.96\textwidth]{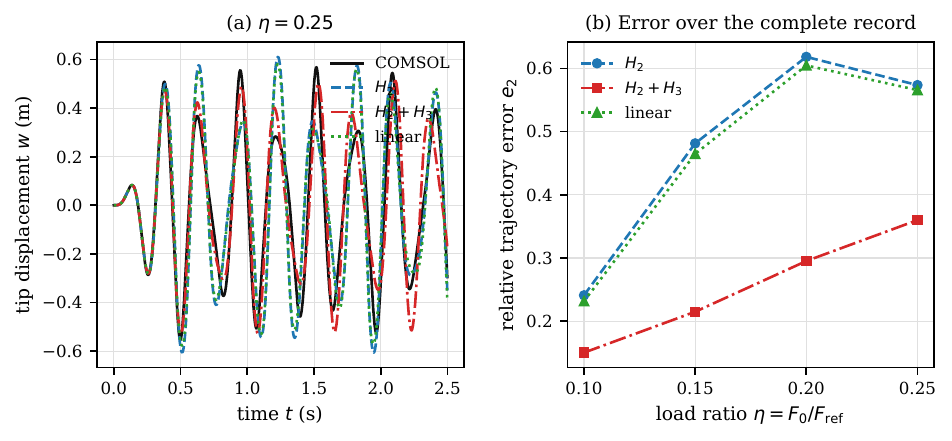}
 \caption{Cantilever comparison from the recorded time histories. (a) Tip displacement at $\eta=0.25$. (b) Relative trajectory error $e_2$ defined by Eq.~\eqref{eq:beam_errors}. The label $H_2+H_3$ denotes Eq.~\eqref{eq:H23}, not an additional retained vibration mode.}
 \label{fig:beam_comparison}
\end{figure}

\begin{table}[htbp]
\caption{Relative errors for the four completed cantilever cases, in percent.}
\label{tab:beam_errors}
\centering
\setlength{\tabcolsep}{5.2pt}
\begin{tabular}{@{}c rrr rrr@{}}
\toprule
& \multicolumn{3}{c}{$100e_2$} & \multicolumn{3}{c}{$100e_A$} \\
\cmidrule(lr){2-4}\cmidrule(l){5-7}
$\eta$ & $H_2$ & $H_2+H_3$ & Linear & $H_2$ & $H_2+H_3$ & Linear \\
\midrule
0.10 & 24.10 & 15.03 & 23.16 & 1.00 & 0.92 & 0.17 \\
0.15 & 48.14 & 21.50 & 46.47 & 3.79 & 1.76 & 1.85 \\
0.20 & 61.83 & 29.52 & 60.47 & 6.61 & 3.52 & 3.12 \\
0.25 & 57.36 & 36.00 & 56.57 & 7.50 & 6.80 & 2.11 \\
\midrule
Mean & 47.85 & 25.51 & 46.67 & 4.73 & 3.25 & 1.81 \\
\bottomrule
\end{tabular}
\end{table}

The $H_2+H_3$ route reduces $e_2$ relative to $H_2$ in every completed case and lowers the mean from $0.479$ to $0.255$, a $46.7\%$ reduction. It also lowers the mean peak error from $4.73\%$ to $3.25\%$. The linear model happens to give the smallest peak-amplitude error for these four cases, while its full-record error remains close to $H_2$. Reporting both measures avoids treating peak agreement as agreement of the transient trajectory. The remaining $15$--$36\%$ full-record errors of the cubic route are substantial. The comparison therefore demonstrates the influence and consistent dynamic use of the cubic static correction for this fixture; it does not establish high predictive accuracy. Broader claims require reference-solver convergence, additional geometries, load spectra, damping models, and experimental or independently verified data.

\FloatBarrier

\section{Error Taxonomy, Scope, and Extension}\label{sec:scope}

The formulation separates errors according to the map that creates them.
\begin{enumerate}
 \item \emph{Attachment-model error} measures whether a beam section and fixed offset adequately represent the actual material motion of the lifting surface. It is a modeling error in $\mathcal A$.
 \item \emph{Structural representation error} includes beam idealization, modal truncation, static-manifold truncation, and local chart truncation in $d$ and $\sigma$.
 \item \emph{Aerodynamic discretization error} includes panel resolution, wake approximation, collocation choice, circulation differentiation, and the potential-flow assumptions.
 \item \emph{Equivalent-load error} contains pressure quadrature and the residual couple $c_{\mathrm{res}}$ in Eq.~\eqref{eq:residual_couple}.
 \item \emph{Station-proxy error} is the additional moment in Eq.~\eqref{eq:station_error}; boundary-condition proxy error is Eq.~\eqref{eq:rhs_proxy_error}.
 \item \emph{Dynamic-reduction error} arises from truncating wake states or other dynamic coordinates after the mechanism-relevant output has been selected.
\end{enumerate}
The classification permits each contribution to be tested independently. A total generalized-force discrepancy can be decomposed schematically as
\begin{equation}
 \Delta Q=\Delta Q_{\mathcal A}+\Delta Q_{\mathrm{str}}
 +\Delta Q_{\mathrm{aero}}+\Delta Q_{\mathrm{eq}}
 +\Delta Q_{\mathrm{stat}}+\Delta Q_{\mathrm{red}},
 \label{eq:error_decomposition}
\end{equation}
although rigorous bounds depend on the regularity and stability constants of the selected models.

The fixed material attachment is well suited to skins, ribs, or lifting surfaces whose local motion is prescribed relative to a supporting beam section. A general nonmatching-mesh interpolation can be preferable when the surface motion cannot be represented by a section and offset, or when a three-dimensional structural discretization is authoritative. The common-source property in Proposition~\ref{prop:common_source} is a consistency statement under the present attachment hypothesis; it does not establish universal superiority over radial-basis or conservative interpolation schemes.

For continuous morphing or sliding attachment, let $\psi=\psi(u,v,\alpha(t))$ depend on a morphing parameter $\alpha$. Equation~\eqref{eq:surface_velocity} then acquires the transport term
\begin{equation}
 \dot X_{\mathrm{rel}}=X_{,s}s_{,\alpha}\dot\alpha
 +X_{,\upsilon}\upsilon_{,\alpha}\dot\alpha
 +R_S\rho_{,\alpha}\dot\alpha.
 \label{eq:sliding_attachment}
\end{equation}
The same term must appear in the virtual map and its cotangent pullback. Patch transitions also require compatible surface orientation and a rule for material-coordinate continuity. These extensions retain the compositional structure while enlarging the attachment state.

The present evidence is structural and interface level. It does not include a complete nonlinear aircraft time simulation, an aerodynamic mesh-convergence study, or experimental validation. Data-driven aeroelastic models may complement this framework when sufficiently rich observations are available \cite{Fonzi2020}; their states and outputs can be coupled through the same declared generalized-force channels.

\section{Conclusions}

A compositional aeroelastic formulation has been organized around the assumed attachment between a beam material section and an aerodynamic surface. The key structural distinction is that the exact relative log $d$ supplies element strain and potential energy, while the reference-anchored section coordinate $\sigma$ supplies the internal geometry used to generate the attached surface. The finite reference geometry is retained exactly in their local expansions. Because the component root pose is independent of the material coordinate, the elastic strain energy of component $i$ reduces to $U_i(q_i)$; upstream motion enters inertia, external work, and boundary-port assembly through the component graph.

The attachment map generates surface points, tangents, normals, velocities, and force Jacobians from a common section field. Aerodynamic boundary conditions and pressure one-forms can therefore share the same geometry and remain paired under virtual work. Collocation, pressure, load, and station-proxy sites were defined as separate discretization choices, yielding explicit residual-couple, station-moment, and right-hand-side proxy errors. Euler--Poincar\'e beam balance and graph cotangent pullback provide the structural and multibody backbone, while wake-state reduction preserves direct and input-rate channels.

The exact-reference relative-log data exhibit the expected $N_d+1$ convergence order. In the nonlinear cantilever comparison, adding the cubic static-manifold term reduces full-record displacement error in every completed load case relative to the quadratic manifold. These results support the consistency of the structural construction and provide quantitative checks for future coupled studies. The next validation stage should combine aerodynamic grid refinement, attachment-model comparisons, and experimental or independently verified aeroelastic responses across jointed and continuously morphing configurations.

\appendix
\renewcommand{\thesubsection}{\thesection.\arabic{subsection}}
\renewcommand{\thesubsubsection}{\thesubsection.\arabic{subsubsection}}
\titleformat{\subsection}
  {\normalsize\bfseries\singlespacing}
  {\thesubsection.\space}{0pt}{#1}[]
\titleformat{\subsubsection}
  {\normalsize\itshape\singlespacing}
  {\thesubsubsection.\space}{0pt}{#1}[]

\section{Weak Euler--Poincar\'e Derivation and Boundary Ports}\label{app:EP}

For a component, let
\begin{equation}
 \mathcal S=\int_{t_0}^{t_1}\int_0^L
 \left[\frac12\xi\T\mathcal M\xi-W(\kappa;s)\right]\dd s\dd t
 +\int_{t_0}^{t_1}\delta W_{\mathrm{ext}}\dd t.
 \label{eq:action_app}
\end{equation}
Using Eq.~\eqref{eq:compatibility}, the internal variation is
\begin{align}
 \delta\mathcal S_{\mathrm{int}}
 =&\int\!\!\int
 \left\{\mu\T(\dot\zeta+\ad_\xi\zeta)
 -n\T(\zeta_{,s}+\ad_\kappa\zeta)\right\}\dd s\dd t.
 \label{eq:weak_before_parts}
\end{align}
Integration by parts in time and material coordinate gives
\begin{align}
 \delta\mathcal S_{\mathrm{int}}
 =&\int\!\!\int
 \left[-\dot\mu+\ad_\xi^*\mu+n_{,s}-\ad_\kappa^*n\right]\T\zeta\dd s\dd t
 \nonumber\\
 &+\int_{t_0}^{t_1}\left[n(0)\T\zeta(0)-n(L)\T\zeta(L)\right]\dd t,
 \label{eq:weak_after_parts}
\end{align}
where $\zeta(t_0)=\zeta(t_1)=0$. Let the external virtual work be
\begin{equation}
 \delta W_{\mathrm{ext}}=\int_0^Lf\T\zeta\dd s
 +w_0\T\zeta(0)+w_L\T\zeta(L).
 \label{eq:external_work_app}
\end{equation}
Arbitrariness of the interior variation gives Eq.~\eqref{eq:EP}; free endpoint variations give
\begin{equation}
 w_0=-n(0),\qquad w_L=n(L).
 \label{eq:boundary_ports}
\end{equation}
These are external wrenches applied to the component. At a connected internal graph port, equal and opposite endpoint work is enforced by the cotangent assembly, and the paired terms cancel in the global virtual work.

For $G=B\overline G$, a left-trivialized section variation can be decomposed as
\begin{equation}
 \zeta=Z(s,q)\zeta_B+Y(s,q)\delta q.
 \label{eq:zeta_ZY}
\end{equation}
Substitution into Eq.~\eqref{eq:weak_after_parts} and transposition gives the root and elastic one-forms. If $U=U(q)$ as in Proposition~\ref{prop:root_invariance}, its direct root-potential one-form vanishes. This does not remove the root inertial one-form or the external and endpoint wrenches.

\section{Reference-Anchored Expansions of \texorpdfstring{$d$ and $\sigma$}{d and sigma}}\label{app:anchored}

\subsection{Exact-reference construction for \texorpdfstring{$d$}{d}}

In the endpoint-relative chart of Eq.~\eqref{eq:endpoint_relative_chart}, define
\begin{equation}
 \chi(q_a,r)=\Log\!\left[
 \Exp(-(A_0q_a)^\wedge)\Exp((A_0q_a+r)^\wedge)
 \right]^\vee .
 \label{eq:chi_d}
\end{equation}
The adjoint identity
\begin{equation}
 \Exp(d_0^\wedge)\Exp(-(A_0q_a)^\wedge)
 =\Exp(-q_a^\wedge)\Exp(d_0^\wedge)
 \label{eq:adjoint_identity_d}
\end{equation}
implies
\begin{equation}
 \Exp(d^\wedge)=\Exp(d_0^\wedge)\Exp(\chi^\wedge).
 \label{eq:d_factorized}
\end{equation}
Introduce the exact-reference increment map
\begin{equation}
 F_{d_0}(e)=\Log\!\left[
 \Exp(-d_0^\wedge)\Exp((d_0+e)^\wedge)
 \right]^\vee .
 \label{eq:F_d0}
\end{equation}
Then $\delta d$ is the local solution of
\begin{equation}
 F_{d_0}(\delta d)=\chi.
 \label{eq:F_inverse}
\end{equation}
Expanding $F_{d_0}$ homogeneously in $e$ gives
\begin{equation}
 F_{d_0}(e)=\sum_{n\geq1}F^{[n]}_{d_0}[e,\ldots,e],
 \qquad F^{[1]}_{d_0}=\dexp_{-d_0}.
 \label{eq:F_series}
\end{equation}
Assuming $F^{[1]}_{d_0}$ is nonsingular in the selected chart, homogeneous series reversion produces
\begin{align}
 \delta d_1&=(F^{[1]}_{d_0})^{-1}\chi_1,
 \label{eq:d_reversion1}\\
 \delta d_n&=-(F^{[1]}_{d_0})^{-1}
 \left[F_{d_0}\left(\sum_{k=1}^{n-1}\delta d_k\right)-\chi\right]_n,
 \quad n\geq2.
 \label{eq:d_reversionn}
\end{align}
The interaction-picture Magnus expansion supplies each $F^{[n]}_{d_0}$ while retaining $d_0$ exactly. The final $\Pi_{\leq N_d}$ is applied in $(q_a,r)$ after reversion.

\subsection{Section-coordinate expansion}

For the section coordinate, Eqs.~\eqref{eq:anchored_defs}--\eqref{eq:sigma_bch} already give the exact split
\begin{equation}
 \sigma=\BCH\!\left(
 \Ad_{\Exp(-\lambda d_0^\wedge)}q_a,
 \chi_R(\lambda d_0,\lambda\delta d)
 \right).
 \label{eq:sigma_app_exact}
\end{equation}
A transparent truncation specification contains three degrees:
\begin{enumerate}
 \item the perturbation degree retained in $\chi_R$;
 \item the total or per-slot degree retained in the small--small BCH of Eq.~\eqref{eq:sigma_app_exact};
 \item the final total degree in the endpoint chart after substitution of $\delta d(q_a,r)$.
\end{enumerate}
For a complete degree-two section card, the small--small terms are
\begin{equation}
 \sigma^{[2]}=\Pi_{\leq2}\left(
 \widehat q_a+\chi_R+\frac12[\widehat q_a,\chi_R]\right).
 \label{eq:sigma_degree2}
\end{equation}
The final projection is essential: a term with one appearance in each BCH slot can still have endpoint degree three if the right slot contains a quadratic anchored increment.

\subsection{Chart validity and differentiation}

The logarithms in Eqs.~\eqref{eq:d_exact} and \eqref{eq:sigma_exact} require a local branch that is continuous around the reference configuration. The reference anchor may be finite, but endpoint perturbations must remain inside that chart. Value convergence alone is insufficient for dynamics. Directional actions should also satisfy
\begin{equation}
 Dd^{[N]}(z)[\delta z]
 =Dd_{\mathrm{exact}}(z)[\delta z]+O(\norm{z}^{N}),
 \label{eq:action_order}
\end{equation}
with the appropriate local scaling. Endpoint identities for $\sigma$ provide additional checks at $\lambda=0$ and $1$.

\section{Static-Manifold and Potential Degree Bookkeeping}\label{app:condensation}

Let the exact reduced relative deformation after substitution of the manifold be
\begin{equation}
 \delta d(q)=d_1(q)+d_2(q)+d_3(q)+\cdots,
 \label{eq:d_homogeneous}
\end{equation}
where $d_k$ is degree $k$. With $U=\tfrac12\delta d\T K_d\delta d$, the degree-$m$ potential is
\begin{equation}
 U_m=\frac12\sum_{j+k=m}d_j\T K_d d_k.
 \label{eq:Um}
\end{equation}
In particular,
\begin{align}
 U_2&=\frac12d_1\T K_d d_1,
 \label{eq:U2}\\
 U_3&=d_1\T K_d d_2,
 \label{eq:U3}\\
 U_4&=d_1\T K_d d_3+\frac12d_2\T K_d d_2,
 \label{eq:U4}\\
 U_6&=d_1\T K_d d_5+d_2\T K_d d_4+
 \frac12d_3\T K_d d_3.
 \label{eq:U6}
\end{align}
Thus, a fourth-order potential requires $d$ through degree three, and a sixth-order potential requires $d$ through degree five. Terms $h_2$ and $h_3$ modify the endpoint map before $d$ is reconstructed. A consistent low-dimensional procedure is:
\begin{enumerate}
 \item use the master--slave relative deformation through degree three to solve Eqs.~\eqref{eq:homological} for $h_2$ and $h_3$;
 \item substitute $x_{H_2}$ into the endpoint chart, reconstruct $d$ through degree three, and retain $\overline U_{H_2}$ through degree four;
 \item substitute $x_{H_2+H_3}$, reconstruct $d$ through degree five, and retain $\overline U_{H_2+H_3}$ through degree six;
 \item verify the common-order identity in Eq.~\eqref{eq:potential_consistency} and the closures in Eq.~\eqref{eq:closure_metric}.
\end{enumerate}
This order matching treats the nonlinear manifold and geometric relative-log expansion as coupled parts of one reduced potential.

\section{Attachment Differentials and Equivalent-Point Errors}\label{app:attachment}

\subsection{Normal variation}

Let $c=X_{,u}\times X_{,v}$ and $n_s=c/\norm c$. A variation of the normal is
\begin{equation}
 \delta n_s=\frac{1}{\norm c}(I-n_sn_s\T)\delta c,
 \label{eq:normal_variation}
\end{equation}
with
\begin{equation}
 \delta c=\delta X_{,u}\times X_{,v}+X_{,u}\times\delta X_{,v}.
 \label{eq:cross_variation}
\end{equation}
Because $\delta X$ is generated from Eq.~\eqref{eq:current_surface}, these derivatives include both section rotation and section translation. The area variation is $\delta(\dd A)=n_s\T\delta c\,\dd u\dd v$.

\subsection{Virtual work of a force and couple}

Let $r=X_L-X_S$. For a section variation $(\delta x_S,\delta\phi_S)$,
\begin{equation}
 \delta X_L=\delta x_S+\delta\phi_S\times r+\delta X_{L,\mathrm{el}},
 \label{eq:point_variation_explicit}
\end{equation}
where $\delta X_{L,\mathrm{el}}$ denotes any additional elastic-offset variation already contained in $J_{X,q}$. The work of $(f,c_L)$ is
\begin{align}
 f\T\delta X_L+c_L\T\delta\phi_S
 =&\ f\T\delta x_S+
 \{r\times f+c_L\}\T\delta\phi_S
 +f\T\delta X_{L,\mathrm{el}}.
 \label{eq:explicit_work}
\end{align}
The first two terms yield Eq.~\eqref{eq:section_wrench}; the last is pulled to elastic coordinates by the point Jacobian. This derivation also shows why moving a force to a station without updating its moment creates Eq.~\eqref{eq:station_error}.

\subsection{Choice of an equivalent point}

Given $f\neq0$ and a target moment $m_O$, a point force can reproduce only the component of $m_O$ perpendicular to $f$. One minimum-norm point offset is
\begin{equation}
 r_L=\frac{f\times m_O}{\norm f^2},
 \label{eq:minimum_norm_point}
\end{equation}
which gives
\begin{equation}
 r_L\times f=m_O-\frac{f f\T}{\norm f^2}m_O.
 \label{eq:point_moment_projection}
\end{equation}
The parallel component
\begin{equation}
 c_{\parallel}=\frac{f f\T}{\norm f^2}m_O
 \label{eq:parallel_couple}
\end{equation}
cannot be represented by relocating the force and must remain as a free couple if moment equivalence is required. For $f=0$, a nonzero panel moment is purely a couple. These relations define the approximation made by a force-only equivalent point.

\section{Three-Segment Folding-Wing Specialization and Mathematical Procedure}\label{app:three_segment}

\subsection{Three-segment graph}

Consider three flexible segments connected serially by two morphing joints. Their roots are
\begin{align}
 B_1&=B_0H_1^{\mathrm{inst}},
 \label{eq:B1_three}\\
 B_2&=P_1C_2(\theta_2)H_2^{\mathrm{inst}},
 \label{eq:B2_three}\\
 B_3&=P_2C_3(\theta_3)H_3^{\mathrm{inst}},
 \label{eq:B3_three}
\end{align}
where $P_i$ is the distal port of segment $i$. Each component has its own $q_i$, $d_i$, $\sigma_i$, $U_i(q_i)$, and attachment map $\mathcal A_i$. No term $U_3(q_3)$ is differentiated directly with respect to $\theta_2$ or $q_2$. Segment-3 inertia and applied aerodynamic loading nevertheless contribute to those upstream coordinates through Eq.~\eqref{eq:graph_pullback}.

Let $w_i^{\mathrm{loc}}$ be the net local root wrench of segment $i$ after its distributed one-forms are pulled to the root. The subtree wrenches satisfy
\begin{align}
 W_3&=w_3^{\mathrm{loc}},
 \label{eq:W3}\\
 W_2&=w_2^{\mathrm{loc}}+\Ad_{H_{2\to3}}^{-\mathsf T}W_3,
 \label{eq:W2}\\
 W_1&=w_1^{\mathrm{loc}}+\Ad_{H_{1\to2}}^{-\mathsf T}W_2,
 \label{eq:W1}
\end{align}
where $H_{i\to j}$ maps from the child-root frame to the parent port frame under the convention of Eq.~\eqref{eq:pairing}. The generalized joint forces are
\begin{equation}
 Q_{\theta_3}=J_3\T W_3,\qquad
 Q_{\theta_2}=J_2\T W_2.
 \label{eq:joint_forces_three}
\end{equation}
The proximal joint accumulates both downstream subtrees. This result follows from topology and covector pullback and does not require a separately prescribed moment rule.

\subsection{Mathematical construction procedure}

The following pseudocode summarizes the dependency order without prescribing a software architecture.
\begin{enumerate}
 \item \emph{Reference definition.} Specify component graphs, beam reference frames and constitutive data, reference lifting surfaces, oriented parameter cells, and attachment maps.
 \item \emph{Element geometry.} For each element, form the exact reference $d_0$, an anchored finite-order $d$, the potential generated by $\delta d$, and a compatible anchored $\sigma$ on required section points.
 \item \emph{Component reduction.} Select $\Phi$ and $\Psi$, solve Eq.~\eqref{eq:homological} to the requested static-manifold degree, reconstruct $d$ at the potential order required by Appendix~\ref{app:condensation}, and verify closure.
 \item \emph{Attached surface.} Derive $\rho$ from Eq.~\eqref{eq:rho_reference}; generate $X$, tangents, normals, velocities, and point Jacobians from Eqs.~\eqref{eq:current_surface}--\eqref{eq:point_jacobian}.
 \item \emph{Aerodynamic one-forms.} Apply Eqs.~\eqref{eq:no_penetration} and \eqref{eq:bernoulli}; integrate pressure loads; retain any residual couples required by Eq.~\eqref{eq:residual_couple}; pull loads through the point and section Jacobians.
 \item \emph{Graph assembly.} Accumulate local one-forms by Eq.~\eqref{eq:graph_pullback}; select mechanism-relevant generalized-force outputs.
 \item \emph{Dynamic reduction.} Reduce wake states using Eqs.~\eqref{eq:balanced_reduction}--\eqref{eq:direct_preservation} while preserving direct and input-rate channels.
 \item \emph{Verification.} Check reference attachment, endpoint identities, normal orientation, value and action convergence, virtual-work duality, static-manifold closure, equivalent-point residuals, and reduced-state diagnostics.
\end{enumerate}

\section*{Data Availability}

The numerical records underlying Figs.~\ref{fig:d_convergence} and \ref{fig:beam_comparison} are available from the corresponding author upon reasonable request.

\section*{Acknowledgments}

The authors thank colleagues in the School of Aeronautic Science and Engineering at Beihang University for discussions on geometrically exact beams, morphing mechanisms, and potential-flow aeroelasticity.

\bibliography{references}

\end{document}